\documentclass[letterpaper, 10pt, conference]{ieeeconf}      

\IEEEoverridecommandlockouts                              
\usepackage{amsmath,graphicx,amsfonts,amssymb,epsfig,subfigure,mathrsfs,mathtools}
\usepackage{arydshln}
\usepackage{blkarray}

\usepackage{color}
\usepackage{multirow}
\usepackage{multicol}
\usepackage{lipsum}
\usepackage{rotating}
\usepackage{graphicx}%
\usepackage{algorithm}
\usepackage{algpseudocode}
\usepackage{cite}
\usepackage{tikz}
\usepackage{textcomp}
\newcommand{\tr}{\text{tr}}
\newcommand{\vect}{\text{vec}}

\newtheorem{theorem}{Theorem}[section]

\newtheorem{lemma}{Lemma}[section]
\newtheorem{proposition}{Proposition}[section]
\newtheorem{remark}{Remark}[section]

\title{\Large \bf
A Switched Dynamical System Framework for Analysis of Massively Parallel Asynchronous Numerical Algorithms
}

\author{Kooktae Lee, Raktim Bhattacharya, and Vijay Gupta
\thanks{Kooktae Lee and Raktim Bhattacharya are with the Department of Aerospace Engineering, Texas A\&M University,
        College Station, TX 77843-3141, USA, {\tt\scriptsize \{animodor,raktim\}@tamu.edu.}
		Vijay Gupta is with the Department of Electrical Engineering, University of Notre Dame,
        Notre Dame, IN 46556, USA, {\tt\scriptsize vgupta2@nd.edu.}
        Kooktae Lee and Raktim Bhattacharya were supported by NSF award \#1349100, and Vijay Gupta was supported by NSF award \#0846631.
        }%
}

\begin{document}
\maketitle
\thispagestyle{empty}
\pagestyle{empty}

\begin{abstract}
In the near future, massively parallel computing systems will be necessary to solve computation intensive applications. The key bottleneck in massively parallel implementation of numerical algorithms is the synchronization of data across processing elements (PEs) after each iteration, which results in significant idle time. Thus, there is a trend towards relaxing the synchronization and adopting an asynchronous model of computation to reduce idle time. However, it is not clear what is the effect of this relaxation on the stability and accuracy of the numerical algorithm.

In this paper we present a new framework to analyze such algorithms. We treat the computation in each PE as a dynamical system and model the asynchrony as stochastic switching. The overall system is then analyzed as a switched dynamical system. However, modeling of massively parallel numerical algorithms as switched dynamical systems results in a very large number of modes, which makes current analysis tools available for such systems computationally intractable. We develop new techniques that circumvent this scalability issue. The framework is presented on a one-dimensional heat equation and the proposed analysis framework is verified by solving the partial differential equation (PDE) in a $\mathtt{nVIDIA\: Tesla^{\scriptsize{TM}}}$ GPU machine, with asynchronous communication between cores.
\end{abstract}
\section{Introduction}

Exascale computing systems will soon be available to study computation intensive applications such as multi-physics multi-scale simulations of natural and engineering systems. Many scientific and practical problems can be described very accurately by ordinary or partial differential equations which may be tightly coupled with long-range correlations. These exascale systems may have $O(10^5 -  10^6)$ processors ranging from multicore processors to symmetric multiprocessors~\cite{fan2004gpu, owens2008gpu, nickolls2008scalable}. Furthermore, such systems are likely to be heterogeneous using both heavily multi-threaded CPUs as well as GPUs. Many challenges must be overcome before exascale systems can be utilized effectively in such applications. One such obstacle is the communication in tightly coupled problems during parallel implementation of any iterative numerical algorithm.  This communication requires massive data movement in turn leading to idle time as the cores need to be synchronized after each time step. 

Recent literature has proposed relaxing these synchronization requirements across the PEs \cite{donzis2014asynchronous}. This potentially eliminates the overhead associated with extreme parallelism and significantly reduces computational time. However, the price to pay is loss of predictability possibly resulting in  calculation errors. Thus, a rigorous analysis of the tradeoff between speed and accuracy is critical. This paper present a framework for quantifying this tradeoff by analyzing the asynchronous numerical algorithm as a switched dynamical system 
\cite{daafouz2002stability,lin2005stability, lee2014optimal, hassibi1999control, xiao2000control, zhang2005new, liu2009stabilization, lee2014acc, lee2015performance, lee2015DNCS}. While there is a large literature for analysis of such systems, these techniques are not applicable to our application. The reason is that due to the large number of PEs, the  switched system model has an extremely large number of modes, which makes the available analysis tools intractable. Key contributions in this paper include new techniques for a) \textit{stability analysis}, or quantification of steady-state error with respect to the synchronous solution; b) \textit{convergence rate analysis} of the expected value of this error;  and c) \textit{probabilistic bounds} on this error. These techniques are developed to be computationally efficient, and avoid the aforementioned scalability issue.

The paper is organized as follows. Section II addresses the problems for the asynchronous numerical algorithm. In section III, we introduce a switched system framework to model the system structure for the asynchronous numerical scheme.
The stability results are presented in section IV, and section V shows the convergence rate analysis. Then, the error analysis in probability is developed in section VI. Section VII demonstrates the usefulness of the proposed method by examples. Finally, section VIII concludes this paper.

\section{Problem Formulation}
\textit{Notation:} The symbol $||\cdot||$ and $||\cdot||_{\infty}$ stand for the Euclidean and infinity norm, respectively. The set of positive integers are denoted by $\mathbb{N}$. Further, $\mathbb{N}_0\triangleq \mathbb{N}\cup \{0\}$. Also, $\lambda(\cdot)$ represents an eigenvalue of a square matrix. In particular, $\lambda_{max}(\cdot)$ and $\lambda_{min}(\cdot)$ denote the largest and the smallest eigenvalue in magnitude, respectively. The symbols $\otimes$, $\text{det}(\cdot)$, $\tr(\cdot)$, and $\vect(\cdot)$ denote Kronecker product, matrix determinant, trace operator, and vectorization operator, respectively. Finally, the symbol \textbf{Pr}($\cdot$) stands for the probability.

In this paper we demonstrate our framework and techniques on the one-dimensional heat equation, given by
\begin{align}
\frac{\partial u}{\partial t} = \alpha\frac{\partial^2 u}{\partial x^2}, \quad t\geq 0, \label{eqn:1}
\end{align}
where $u$ is the time and space-varying state of the temperature, and $t$ and $x$ are continuous time and space respectively. The constant $\alpha>0$ is the thermal diffusivity of the given material. 

The PDE is solved numerically using the finite difference method by Euler explicit scheme, with a forward difference in time and a central difference in space. Thus \eqref{eqn:1} is approximated as 
\begin{align}
\frac{u_{i}(k+1) - u_{i}(k)}{\Delta t} &= \alpha\left(\frac{u_{i+1}(k) - 2u_{i}(k) + u_{i-1}(k)}{\Delta x^2}\right),\label{eqn:pde}
\end{align}
where $k\in\mathbb{N}_0$ is the discrete-time index and $u_i$ is the temperature value at $i^{th}$ grid space point. The symbols $\Delta t$ and $\Delta x$ denote the sampling time and the grid resolution in space, respectively. Further, if we define a constant $ r\triangleq\alpha\frac{\Delta t}{\Delta x^2}$, then \eqref{eqn:pde} can be written as 
\begin{align}
u_{i}(k+1) = ru_{i+1}(k) + (1-2r)u_{i}(k) + ru_{i-1}(k),\label{eqn:sync}
\end{align}
It is important to observe that \eqref{eqn:sync} is a discrete-time linear dynamical system.

\begin{figure}[h!]
\begin{center}
\includegraphics[scale=0.4]{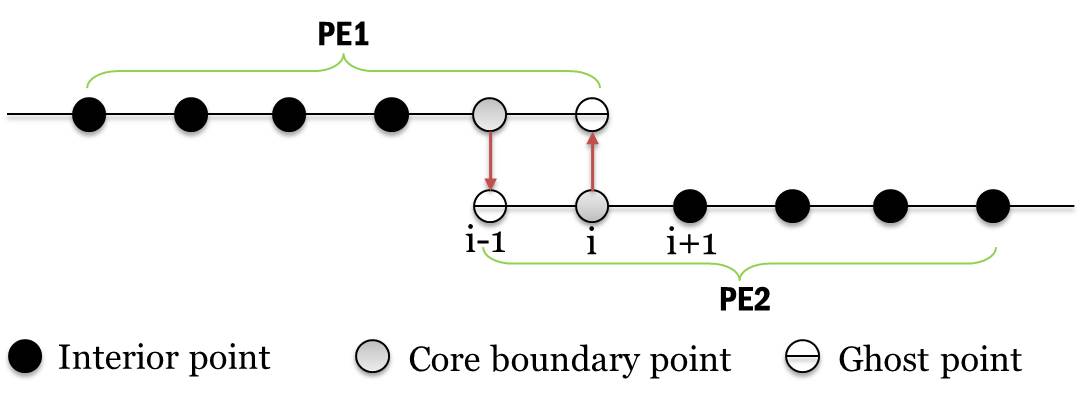}
\caption{Discretized one-dimensional domain with an asynchronous numerical algorithm. the PE denotes a group of grid points, assigned to each core.}\label{fig.0}
\end{center}
\end{figure}
Fig. \ref{fig.0} illustrates the numerical scheme over the discretized 1D spatial domain. A typical \textit{synchronous} parallel implementation of this numerical scheme  assigns several of these grid points to each PE. The updates for the temperature at the grid points assigned to each PE, occur in parallel. However, at every time step $k$, the data associated with the boundary grid points, where the communication is necessary are synchronized, and used to compute $u_i(k+1)$. This synchronization across PEs is slow, especially for massively parallel systems (estimates of idle time due to this synchronization give figures of up to 80\% of the total time taken for the simulation as idle time). Recently, an alternative implementation which is \textit{asynchronous} has been proposed. In this implementation, the updates in a PE occur without waiting for the other PEs to finish and their results to be synchronized. The data update across PEs occurs sporadically and independently. This asynchrony directly affects the update equation for the boundary points, as they depend on the grid points across PEs. For these points, the update is performed with the most recent available value, typically stored in a buffer. The effect of this asynchrony then propagates to other grid points. Within a PE, we assume there is no asynchrony and data is available in a common memory. 

Thus, the asynchronous numerical scheme corresponding to \eqref{eqn:sync} is given by
\begin{align}
u_{i}(k+1) = ru_{i+1}(k^{*}_{i+1}) + (1-2r)u_{i}(k) + ru_{i-1}(k^{*}_{i-1}),\label{eqn:async}
\end{align}
where $k^{*}_{i}\in \{k,k-1,k-2,\hdots,k-q+1\}$, $i=1,2,\hdots,N$, denotes the randomness caused by communication delays between PEs. The subscript $i$ in $k_i^*$ depicts that each grid space point may have different time delays.
The parameter $q$ is the length of a buffer that every core maintains to store data transmitted from the other cores. In this paper, we treat $k^{\ast}_{i}$ as a random variable and thus \eqref{eqn:async} can be considered to be a linear discrete-time dynamical system with stochastic updates. 

Although \eqref{eqn:async} is derived for the 1D heat equation, the treatment above can be developed for any parabolic PDEs. This observation encourages us to consider using tools from dynamical systems to analyze the effect of asynchrony in parallel numerical algorithms. Therefore, the primary goal of this study is to investigate the \textit{stability}, \textit{convergence rate}, and \textit{error probability} of the asynchronous numerical algorithm in the framework of stochastic switched dynamical systems. 

\section{A Switched System Approach}
Let us define the state vector $U_j(k)\in\mathbb{R}^{n} \triangleq [u_1^j(k), u_2^j(k),\hdots,u_n^j(k)]^{\top}$, where $u_i^j(k)$  stands for the $i^{th}$ grid space point in the $j^{th}$ PE and $n$ is the total number of grid points in the $j^{th}$ PE. Therefore, \eqref{eqn:sync} can be compactly written as
\begin{align*}
U(k+1) = AU(k),\quad k\in \mathbb{N}_0,
\end{align*}
where $U(k)\in\mathbb{R}^{Nn}\triangleq [U_1(k)^{\top}, U_2(k)^{\top},\hdots,U_N(k)^{\top}]^{\top}$, $N$ is the total number of PEs, $n$ is the size of the state for each PE, and system matrix $A\in\mathbb{R}^{Nn\times Nn}$ is given by
\small{
\begin{align*}
A &= \begin{bmatrix}
1 & 0 & 0 & \cdots & \cdots & 0\\
r & 1\text{-}2r & r & 0 & \cdots & 0\\
0 & r & 1\text{-}2r & r & \cdots & 0\\
\vdots & & \ddots & \ddots & \ddots & \\
0 & \cdots & & r & 1\text{-}2r& r\\
0 & \cdots & & 0 & 0 & 1
\end{bmatrix}\in\mathbb{R}^{Nn\times Nn}.
\end{align*}
}\normalsize

Note that the first and the last row of $A$ matrix specify the Dirichlet boundary conditions (see pp. 150, \cite{pletcher2012computational}).  i.e., we have the constant in time boundary temperatures for simplicity.
%
%
%
%

Next, we define the augmented state $X(k)\in\mathbb{R}^{Nnq}\triangleq [U(k)^{\top}, U(k-1)^{\top}, \hdots, U(k-q+1)^{\top}]^{\top}$, where, as stated before, $q$ is the buffer length.
For pedagogical simplicity (and without loss of generality), we consider the case with $q=2$ and $N=3$. Further, we let $n=1$, which implies there is only one grid point in each PE. For this particular case, we construct the following matrices,

\footnotesize{
\[W_1 = \left[ \begin{array}{c:c}
1 \quad 0 \quad 0 & 0 \quad 0 \quad 0\\
r \:\: 1\text{-}2r \:\: r & 0 \quad 0 \quad 0\\
0 \quad 0 \quad 1 & 0 \quad 0 \quad 0\\
\hdashline
& \\ 
  $\Large{\underline{I}}$ & $\Large{\underline{0}}$ \\ 
&\\
\end{array} \right],\:
W_2 = \left[ \begin{array}{c:c}
1 \quad 0 \quad 0 & 0 \quad 0 \quad 0\\
0 \:\: 1\text{-}2r \:\: r & r \quad 0 \quad 0\\
0 \quad 0 \quad 1 & 0 \quad 0 \quad 0\\
\hdashline
& \\ 
  $\Large{\underline{I}}$ & $\Large{\underline{0}}$ \\ 
&\\
\end{array} \right],\]
\[W_3 = \left[ \begin{array}{c:c}
1 \quad 0 \quad 0 & 0 \quad 0 \quad 0\\
r \:\: 1\text{-}2r \:\: 0 & 0 \quad 0 \quad r\\
0 \quad 0 \quad 1 & 0 \quad 0 \quad 0\\
\hdashline
& \\ 
  $\Large{\underline{I}}$ & $\Large{\underline{0}}$ \\ 
&\\
\end{array} \right],\:
W_4 = \left[ \begin{array}{c:c}
1 \quad 0 \quad 0 & 0 \quad 0 \quad 0\\
0 \:\: 1\text{-}2r \:\: 0 & r \quad 0 \quad r\\
0 \quad 0 \quad 1 & 0 \quad 0 \quad 0\\
\hdashline
& \\ 
  $\Large{\underline{I}}$ & $\Large{\underline{0}}$ \\ 
&\\
\end{array} \right],
\]
}\normalsize
where \underline{I} $\in\mathbb{R}^{Nn\times Nn}$ and \underline{0} $\in\mathbb{R}^{Nn\times Nn}$ are the identity and the zero matrices with appropriate dimensions. As in \cite{donzis2014asynchronous}, we assume that the condition $0 < r \leq 0.5$ holds from now on.
The asynchronous numerical scheme can then be written as a switched system
\begin{align}
X(k+1) = W_{\sigma_{k}}X(k),\:\:\:\sigma_k \in \{1,2,\hdots,m\},\:k\in\mathbb{N}_0,\label{eqn:switched}
\end{align}
where the matrices $W_{\sigma_{k}}\in\mathbb{R}^{Nnq\times Nnq}$, are the subsystem dynamics. In general, the total number of switching modes is $m=q^{2(N-2)}$ that is obtained by considering all cases to distribute every components $r$ in $W_1$ matrix, where the number of $r$ in $W_1$ is given by $2(N-2)$, into $q$ numbers of zero block matrix as in the above example. Therefore, the number of modes increase exponentially with the number of PEs, which is quite large for massively parallel systems. 

At every time step, the numerical scheme evolves using one of the $m$ modes, which depends on the variable $k^{\ast}_{i}$. In this paper, we model the variable $k^{\ast}_{i}$ as a random variable that evolves in an independently and identically distributed (i.i.d.) fashion in time, and independently from one core to the next. Hence, we let $\pi_j$ be the modal probability for $W_j$ which is assumed to be stationary in time. Let $\Pi\triangleq\{\pi_1,\pi_2,\hdots,\pi_m\}$, be the switching probabilities such that $0\leq\pi_j\leq 1$, $\forall j$ and $\sum_{j=1}^{m}\pi_j=1$. The system in \eqref{eqn:switched} is thus an i.i.d jump linear system, which is a simpler case of the more well-known Markovian jump linear systems \cite{lee2015performance}. Even though the analysis theory for such systems is well developed, the existing tools are not suitable for our application because of the extremely large number of modes, particularly when $N$ is large. Thus, we now develop an analysis theory for the i.i.d. jump linear systems which scales better with respect to the number of modes.

\section{Stability}
The first requirement is that of convergence of~\eqref{eqn:switched}. Because of the Dirichlet boundary conditions, we expect the temperature to converge to a constant value for every grid point. We proceed to analyze the conditions for convergence (or stability) of the system. To this end, we may try to use the infinity norm and apply the sub-multiplicative property to obtain $|| X(k+1)||_{\infty} = || W_{\sigma_{k}}X(k)||_{\infty} 
\leq$ $|| W_{\sigma_{k}}||_{\infty}|| X(k)||_{\infty} = || X(k)||_{\infty}
$, where the last equality holds since we have $|| W_j ||_{\infty} = 1$, $\forall j$.
This can be written as
\begin{equation}
\displaystyle\frac{\parallel X(k+1)\parallel_{\infty}}{\parallel X(k)\parallel_{\infty}}\leq 1.\label{eqn:7}
\end{equation}
The above result only shows that the solution from the asynchronous algorithm is \textit{marginally} stable and we are unable to determine the steady-state solution.

In fact, we can show that the asynchronous scheme also attains the same steady-state value as the synchronous scheme, regardless of the specific realization of $\{\sigma_k\}$. Using spectral decomposition, the matrices $W_j$ can be expressed in terms of the eigenvalues and corresponding eigenvectors as
\begin{align}
W_j\in\mathbb{R}^{Nnq\times Nnq} = \sum_{i=1}^{Nnq}\lambda_i^jv_i^js_i^j,\quad j=\{1,2,\hdots,m\},\label{eqn:spectral_decompostion}
\end{align}
where $\lambda_i^j\in\mathbb{R}$, $v_i^j\in\mathbb{R}^{Nnq\times 1}$, and $s_i^j\in\mathbb{R}^{1\times Nnq}$ denote the eigenvalues, right eigenvectors, and left eigenvectors of $W_j$, respectively.

Since  $\displaystyle\max_{i}|\lambda_i^j|\leq ||W_j||_{\infty} = 1$, $\forall j$, the spectral radius of $W_j$, $j=1,2,\hdots,m$, is less than or equal to $1$. Therefore, we may order the eigenvalues as $1 \geq |\lambda_1^j| \geq |\lambda_2^j| > \cdots \geq |\lambda_{Nnq}^j| \geq 0$.
It can be shown that all $W_j$ have two eigenvalues with value $1$, irrespective of the size of $q$ and $N$. Therefore, the eigenvalues for $W_j$ are ordered as $1 = |\lambda_1^j| = |\lambda_2^j| > |\lambda_3^j| \geq \cdots \geq |\lambda_{Nnq}^j| \geq 0$.

Moreover, the left and right eigenvectors for eigenvalues equal to 1 are common eigenvectors for all matrices $W_j$, $j=1,2,\hdots,m$.
These common left and right eigenvectors are

1) \textbf{Left eigenvectors:}
\begin{align}
s_1 = [1,0,\cdots,0 \:\:, \:\:\textbf{\textsl{0}}\:\: , \cdots , \;\;\textbf{\textsl{0}}\;\;]\in\mathbb{R}^{1\times Nnq},\label{s1}\\
s_2 = [0,\cdots,0,1 \:\:, \:\:\textbf{\textsl{0}}\:\: , \cdots , \;\;\textbf{\textsl{0}}\;\;]\in\mathbb{R}^{1\times Nnq},\label{s2}
\end{align}

2) \textbf{Right eigenvectors:}
\begin{align}
v_1 = [\mu_1,\mu_1,\cdots,\mu_1]^{\top}\in\mathbb{R}^{Nnq\times 1},\label{v1}\\
v_2 = [\mu_2,\mu_2,\cdots,\mu_2]^{\top}\in\mathbb{R}^{Nnq\times 1},\label{v2}
\end{align}
where $\textbf{\textsl{0}}\in\mathbb{R}^{1\times Nn}$ denotes a row vector with all zero elements, and
$\mu_1 \triangleq [1,\frac{Nn-2}{Nn-1},\cdots, \frac{Nn-j}{Nn-1}, \cdots , \frac{1}{Nn-1},0]\in\mathbb{R}^{1\times Nn}$,
$\mu_2 \triangleq [0,\frac{1}{Nn-1},\frac{2}{Nn-1},\cdots, \frac{j-1}{Nn-1}, \cdots , \frac{Nn-2}{Nn-1},1]\in\mathbb{R}^{1\times Nn}$, $j=1,2,\hdots,Nn$.

Notice that we have $W_jv_i = v_i$ and $s_iW_j = s_i$, $i=1,2$, $\forall j$. Then, the steady-state value for the asynchronous scheme is given by the following result.

\begin{proposition}
Consider the i.i.d. jump linear system in \eqref{eqn:switched} with subsystem matrices $W_j$, $j=1,2,\hdots,m$ and a stationary switching probability $\Pi$. For a given initial condition $X(0)$, if we define $\Psi\triangleq v_1s_1 + v_2s_2$, where $v_i$ and $s_i$, $i=1,2$, are given in \eqref{s1}--\eqref{v2}, then, the steady-state value $X_{ss}$ has the following form:
\begin{align*}
X_{ss}\triangleq \lim_{k\rightarrow\infty}X(k)=\Psi X(0),
\end{align*}
irrespective of the switching sequence $\{\sigma_k\}$.
\end{proposition}

\begin{proof}
Let the eigenvalues of $W_j$ be ordered in magnitude by $ 1 = |\lambda_1^j| = |\lambda_2^j| > |\lambda_3^j| \geq \cdots \geq |\lambda_{Nnq}^j| \geq 0$.
Also, let $v_i^j$ and $s_i^j$ be the right and left eigenvector corresponding to $\lambda_i^j$, respectively.
Using the spectral decomposition, $W_j$ can be alternatively expressed by $W_j = \sum_{i=1}^{Nnq}\lambda_i^{j}v_i^{j}s_i^{j} = \Psi + \sum_{\lambda_i^j\neq 1}f^j(i)$, where $\Psi\triangleq v_1s_1 + v_2s_2$ and $f^j(i) \triangleq \lambda_i^{j}v_i^{j}s_i^{j}$.

Then, starting with $X(0)$, the realization of the switching sequence $\sigma_k$ results in
\begin{align*}
&X(k) = W_{\sigma_{k-1}}W_{\sigma_{k-2}}\cdots W_{\sigma_{1}}W_{\sigma_{0}}X(0)\\
&= \Big(\Psi + \sum_{\lambda_i^{\sigma_{k-1}}\neq 1}f^{\sigma_{k-1}}(i)\Big)\cdots
\Big(\Psi + \sum_{\lambda_i^{\sigma_{0}}\neq 1}f^{\sigma_0}(i)\Big)X(0)\\
&= \Big(\Psi^k + g(k)\Big)X(0),
\end{align*}
where in above equation, $g(k)$ represents all the other multiplication terms except $\Psi^k$ term.
Note that $g(k)$ is formed by the product of $\lambda_i^j$, where $0\leq |\lambda_i^j| < 1$, $\forall i>2$, $\forall j$. Consequently, if $k\rightarrow\infty$, then $g(k)$ is asymptotically convergent to zero since the \textit{infinite} number of multiplication of the term $\lambda_i^j$, $\forall i>2$, converges to zero.
Therefore, we have
\begin{align*}
X_{ss} = \lim_{k\rightarrow\infty}X(k) = \lim_{k\rightarrow\infty}\Psi^kX(0)=\Psi X(0).
\end{align*}
The last equality in above equation holds because $\Psi^k = \Psi^{k-1}=\cdots=\Psi$, $\forall k\in\mathbb{N}$.
\end{proof}

\section{Convergence rate}
In this section, we investigate how fast the expected value of the state converges to the steady-state $X_{ss}$ by analyzing the transient behavior of the asynchronous algorithm. Let us define a new state variable $e(k)\triangleq X(k) - X_{ss}$. The expected value of $e(k)$ is given by
$\bar{e}(k) \triangleq \mathbb{E}[X(k)-X_{ss}] = \mathbb{E}[X(k)] - X_{ss} = \bar{X}(k) - X_{ss}$, where $\bar{X}(k)\triangleq \mathbb{E}[X(k)]$.
Therefore, the convergence rate of $||\bar{e}(k)||$ will provide bound for the convergence rate of $||\bar{X}(k)-X_{ss}||$.

To obtain an upper bound for the convergence rate of $||\bar{e}(k)||$, we use the following matrix transformation. As described in \eqref{eqn:spectral_decompostion}, each modal matrix $W_j$ can be alternatively expressed by $W_j = \sum_{i=1}^{Nnq}\lambda_i^jv_i^js_i^j$, where $\lambda_i^j$, $v_i^j$, and $s_i^j$ denote the eigenvalues, right and, respectively, left eigenvectors for $W_j$.
If we define the transformed matrix $\tilde{W}_j \triangleq W_j - \sum_{\lambda_i^1=1}\lambda_i^jv_i^js_i^j = W_j - \Psi = \sum_{\lambda_i^j\neq 1}\lambda_i^jv_i^js_i^j$, then the modal dynamics with the corresponding state $e_j(k)$, is given by
\begin{align}
e_j(k+1) = \tilde{W}_je_j(k),\quad j=\{1,2,\hdots,m\}, \: k\in\mathbb{N}_0.\label{eqn:8-1}
\end{align}
Moreover, as in \eqref{eqn:switched}, the error state $e(k) = X(k) - X_{ss}$, is governed by
\begin{align}
e(k+1) = \tilde{W}_{\sigma_k}e(k),\quad \sigma_k\in\{1,2,\hdots,m\},\:k\in\mathbb{N}_0.\label{eqn:8}
\end{align}

The system in \eqref{eqn:8} is also a switched linear system. The transformed matrix $\tilde{W_j}$ are the modes of the error dynamics. Generally, it is difficult to estimate the convergence rate of the ensemble with stochastic jumps. Previous works \cite{laffey2007tensor,shorten2003result,liberzon2003switching,gurvits2007stability} have used the common Lyapunov function approaches, to analyze stability and the convergence rate. However, the existence of a common Lyapunov function is the \textit{only sufficient condition} for the system stability, and hence there may not exist a common Lyapunov function for the asynchronous algorithm. Moreover, extremely large values of $m$ make it very difficult to test every conditions for the existence of such a common Lyapunov function. For this reason, we bound the convergence rate of $\bar{e}(k)$, instead of bounding $e(k)$ directly.

\begin{lemma}\label{lemma:5.1}
Consider an i.i.d. jump linear system given by \eqref{eqn:8} with the switching probability $\Pi=\{\pi_1,\pi_2,\hdots,\pi_m\}$. 
If the initial state $e(0)$ is given and has no uncertainty, the expected value of $e(k)$ is updated by
\begin{align}
\bar{e}(k)&\triangleq\mathbb{E}[e(k)] = \Lambda^ke(0)\quad \text{or}\quad
\bar{e}(k+1) = \Lambda \bar{e}(k),\label{eqn:11-1}
\end{align}
\end{lemma}
where $\displaystyle\Lambda\triangleq \sum_{i=1}^{m}\pi_i\tilde{W}_i$.

\begin{proof}
For an i.i.d. jump process with a given deterministic initial error $e(0)$, we have
\begin{align*}
\mathbb{E}[e(k)] &= \mathbb{E}[\tilde{W}_{\sigma_{k\text{--}1}}e(k\text{--}1)]\\
&=\mathbb{E}[\tilde{W}_{\sigma_{k\text{--}1}}\tilde{W}_{\sigma_{k\text{--}2}}\hdots \tilde{W}_{\sigma_{1}}\tilde{W}_{\sigma_{0}}e(0)]\\
&= \underbrace{\mathbb{E}[\tilde{W}_{\sigma_{k\text{--}1}}]}_{=\Lambda}\hdots\underbrace{\mathbb{E}[\tilde{W}_{\sigma_{1}}]}_{=\Lambda}\underbrace{\mathbb{E}[\tilde{W}_{\sigma_{0}}]}_{=\Lambda}e(0)=\Lambda^{k} e(0).
\end{align*}
\end{proof}

Since the matrix $\Lambda$ is given by $\Lambda = \sum_{i=1}^{m}\pi_i\tilde{W}_i$, the computation of $\Lambda$ requires all matrices $W_j$, $j=1,2,\hdots,m$. As pointed out earlier, this calculation is intractable due to the extremely large number of the switching modes $m$. Therefore, instead of using \eqref{eqn:11-1}, we provide a computationally efficient method to bound $||\bar{e}(k)||$ through a Lyapunov theorem.

Consider a discrete-time Lyapunov function $V(k) = \bar{e}(k)^{\top}P\bar{e}(k)^{}$, where $P$ is a positive definite matrix.
Since it is shown that the original state $X(k)$ is convergent to the unique steady-state $X_{ss}$ as $k\rightarrow \infty$ irrespective of $\{\sigma_k\}$, the expected error $\bar{e}(k) \triangleq \bar{X}(k) - X_{ss}$ is asymptotically stable. Therefore, one can employ the \textit{Converse Lyapunov Theorem}\cite{lin2009stability}, which guarantees the existence of a positive definite matrix $P$, satisfying the following linear matrix inequality (LMI) condition $\Lambda^{\top}P\Lambda^{} - P< -Q$, where $Q$ is some positive definite matrix. The matrix inequality can be interpreted in the sense of positive definiteness. (i.e., $A>B$ means the matrix $A-B$ is positive definite.) 
Then, the above LMI condition results in
$\Delta V(k) = V(k+1) - V(k) = \bar{e}(k)^{\top}(\Lambda^{\top}P\Lambda^{} - P)\bar{e}(k)^{}
< -\bar{e}(k)^{\top}Q\bar{e}(k)
\leq -\lambda_{min}(Q)$ $\parallel \bar{e}(k)\parallel^2$.
Also, the Lyapunov function $V(k)$ satisfies
\begin{align*}
\lambda_{min}(P)\parallel \bar{e}(k)\parallel^2 \:\leq\: V(k) \:\leq \lambda_{max}(P)\parallel \bar{e}(k)\parallel^2,
\end{align*}
resulting in $\displaystyle-\parallel \bar{e}(k)\parallel^2 \leq -\frac{V(k)}{\lambda_{max}(P)}$. Therefore, we have
\begin{eqnarray}
\Delta V(k) < -\lambda_{min}(Q)\parallel \bar{e}(k)\parallel^2\leq -\frac{\lambda_{min}(Q)}{\lambda_{max}(P)}V(k).\nonumber \\
\Rightarrow V(k+1) < \Big(1-\frac{\lambda_{min}(Q)}{\lambda_{max}(P)}\Big)V(k). \label{eqn_Vk+1<Vk}
\end{eqnarray}
Hence, $\parallel \bar{e}(k)\parallel$ is bounded by a following equation:
\begin{align}
\parallel \bar{e}(k)\parallel^2 < K\bigg(1-\dfrac{\lambda_{min}(Q)}{\lambda_{max}(P)}\bigg)^k\parallel e(0)\parallel^2, \label{eqn_M(1-a)}
\end{align}
where $K>0$ is some constant.

Next, we bound the convergence rate for $||\bar{e}(k)||$ by using the result in $\eqref{eqn_M(1-a)}$ as follows.\\

\begin{proposition}\label{proposition:a<abar => boundedness}
For a \textit{stable} i.i.d. jump linear system \eqref{eqn:8} with a stationary switching probability $\Pi$, consider a Lyapunov candidate function for the state $\bar{e}$, given by $V \triangleq \bar{e}^{\top}P\bar{e}$, where $P$ is a positive definite matrix. In addition, a Lyapunov candidate function for \eqref{eqn:8-1} is given by $V_j \triangleq e_j^{\top}P_je_j^{}$, $j=1,2,\hdots,m$, where $P_j$ is a positive definite matrix.
According to the Converse Lyapunov Theorem, there exist $P_j>0$ and $P>0$ such that $\tilde{W}_j^{\top}P_j\tilde{W}_j^{} - P_j < -Q_j, \: j=1,2,\hdots,m$ and $\Lambda^{\top}P\Lambda^{} - P < -Q$, where $Q_j$ and $Q$ are any positive definite matrices. 
Then, with a particular choice of these matrices, we assume that $P_j$ and $P$ satisfy the following conditions:
\begin{eqnarray}
&\tilde{W}_j^{\top}P_j\tilde{W}_j^{} - P_j = -I, &\quad j=1,2,\hdots,m,\label{eqn:14}\\
&\Lambda^{\top}P\Lambda^{} - P \leq -\varepsilon_j I, &\quad \text{for some }j,\label{Abar'PbarAbar-Pbar < -cPbar}
\end{eqnarray}
where $\displaystyle \varepsilon_j \triangleq \dfrac{\lambda_{max}(P)}{\lambda_{max}(P_j)}> 0$, $\tilde{W}_j$ are the modal matrices in \eqref{eqn:8-1}, and $\displaystyle\Lambda \triangleq \sum_{j=1}^{m}\pi_j\tilde{W}_j$.

Then, $||\bar{e}(k)||^2$ is bounded by
\begin{align}
\parallel \bar{e}(k)\parallel^2 < K\bigg(1-\dfrac{1}{\lambda_{max}(P_j)}\bigg)^k\parallel e(0)\parallel^2, \label{eqn_M_j(1-a_j)}
\end{align}
\end{proposition}
where $K>0$ is some constant.\\

\begin{proof}
By applying the result in \eqref{eqn_M(1-a)} into \eqref{Abar'PbarAbar-Pbar < -cPbar}, we have
\begin{align*}
\parallel \bar{e}(k)\parallel^2 &< K\bigg(1-\frac{\lambda_{min}(\varepsilon_j I)}{\lambda_{max}(P)}\bigg)^k\parallel e(0)\parallel^2\\
&= K\bigg(1-\frac{\varepsilon_j}{\lambda_{max}(P)}\bigg)^k\parallel e(0)\parallel^2\\
&= K\bigg(1-\frac{1}{\lambda_{max}(P_j)}\bigg)^k\parallel e(0)\parallel^2.
\end{align*}
The last equality in above equation holds by the definition of $\varepsilon_j$.
\end{proof}

Proposition \ref{proposition:a<abar => boundedness} says that we can always guarantee the bound for $||\bar{e}(k)||$ if \eqref{Abar'PbarAbar-Pbar < -cPbar} holds. Consequently, the existence of such a $P$, satisfying \eqref{Abar'PbarAbar-Pbar < -cPbar} is the major concern in order to guarantee the bound $||\bar{e}(k)||$. The following lemma and theorem can be used to prove the existence of such a $P$.

\begin{lemma}\label{lemma:5.2}
Suppose that $P_j$ is a positive definite matrix, satisfying \eqref{eqn:14}. Then, the largest eigenvalue of $P_j$ is strictly greater than $1$ for all $j$, i.e., $\lambda_{max}(P_j) > 1$, $\forall j$.
\end{lemma}

\begin{proof}
From \eqref{eqn:14}, $P_j = \tilde{W}_j^{\top}P_j\tilde{W}_j + I$, $\forall j$. Then, with the eigenvectors $y\in\mathbb{R}^{Nnq}$ of $P_j$, the largest eigenvalue of $P_j$ is given by its definition as follows:
\begin{align*}
\lambda_{max}(P_j) &= \lambda_{max}(\tilde{W}_j^{\top}P_j\tilde{W}_j + I)\\
&= \max_{\substack{y\\||y||^2=1}}y^{\top}(\tilde{W}_j^{\top}P_j\tilde{W}_j + I)y\\
&= \max_{\substack{y\\||y||^2=1}}\left(y^{\top}\tilde{W}_j^{\top}P_j\tilde{W}_jy\right) + \underbrace{y^{\top}y}_{=||y||^2=1}
\end{align*}
Since $P_j$ is a positive definite matrix, $\tilde{W}_j^{\top}P_j\tilde{W}_j$ becomes a positive semi-definite matrix at least. 
Then, the scalar term $y^{\top}\tilde{W}_j^{\top}P_j\tilde{W}_jy$ cannot be zero unless $\tilde{W}_j^{\top}P_j\tilde{W}_j$ is a zero matrix or a triangular matrix with zero diagonal components, which is not the case. Hence, it is guaranteed that $y^{\top}\tilde{W}_j^{\top}P_j\tilde{W}_jy > 0$, implying $\lambda_{max}(P_j) > 1$, $\forall j$.
\end{proof}

\begin{theorem}\label{Theorem:5.1}
Consider Lyapunov functions for \eqref{eqn:8-1} and \eqref{eqn:8} given by $V_j \triangleq e_j^{\top}P_je_j$, $j=1,2,\hdots,m$, and $V\triangleq \bar{e}^{\top}P\bar{e}$, respectively, where the matrices $P_j>0,\forall j$ and $P>0$.
By the Converse Lyapunov Theorem, we assume that the matrices $P_j$, $\forall j$, satisfies
the condition \eqref{eqn:14}.

Then, there exists a positive definite matrix $P$ such that
\begin{align}
\Lambda^{\top}P\Lambda^{} - P \leq -\varepsilon_j I, \quad \text{for some }j,\label{eqn:20}
\end{align}
where $\displaystyle \varepsilon_j \triangleq \dfrac{\lambda_{max}(P)}{\lambda_{max}(P_j)}> 0$.
\end{theorem}

\begin{proof}
We prove by contradiction. Suppose that there exist no such $P>0$, satisfying \eqref{eqn:20}, which is equivalent to that for \textit{\textbf{all}} matrices $P>0$, the inequality $\Lambda^{\top}P\Lambda^{} - P > -\varepsilon_j I$ holds $\forall j$.
The above inequality can be interpreted in the quadratic sense.
In other words, for any non-zero vector $v$ that has a proper dimension, the following condition holds:
\begin{align}
v^{\top}\left(\Lambda^{\top}P\Lambda^{} - P  + \varepsilon_j I\right)v > 0, \: \forall j\label{eqn:18}
\end{align}
As a particular choice of $v$, we let the vector $v$ be the eigenvector of the matrix $\Lambda$, i.e., $\Lambda v = \lambda\Lambda$, where $\lambda$ is the eigenvalue of $\Lambda$. Since \eqref{eqn:18} holds for any matrix $P>0$, we let $P=I$, which results in $\varepsilon_j = \dfrac{\lambda_{max}(I)}{\lambda_{max}(P_j)} = \dfrac{1}{\lambda_{max}(P_j)}$. Hence, we have
\begin{align*}
0 &< v^{\top}\left(\Lambda^{\top}\Lambda^{} - I  + \dfrac{1}{\lambda_{max}(P_j)} I\right)v\\
&= (\underbrace{\Lambda v}_{=\lambda v})^{\top}(\underbrace{\Lambda v}_{=\lambda v}) - ||v||^2 + \dfrac{1}{\lambda_{max}(P_j)}||v||^2\\
&= \left(\lambda^2 - 1 + \dfrac{1}{\lambda_{max}(P_j)}\right)||v||^2, \quad \forall j.
\end{align*}

From the structure of the matrix $\Lambda$, it can be shown that $\text{det}(\Lambda)=0$. Therefore, one of the eigenvalues $\lambda$ is zero.
Moreover, Lemma \ref{lemma:5.2} states that $\dfrac{1}{\lambda_{max}(P_j)} < 1$, $\forall j$. As a consequence, with $\lambda
=0$, we have
\begin{align*}
0 &< \underbrace{\left(- 1 + \dfrac{1}{\lambda_{max}(P_j)}\right)}_{ < 0}\underbrace{||v||^2}_{>0} < 0, \quad \forall j.
\end{align*}
which is a \textit{contradiction}.
\end{proof}

\begin{remark}
Proposition \ref{proposition:a<abar => boundedness} provides a very efficient way to bound the convergence rate for $||\bar{e}(k)||$. According to the proposed methods, it is unnecessary to compute the matrix $\Lambda$ and to keep all matrices $W_j$, $j=1,2,\hdots,m$ since $||\bar{e}(k)||$ is bounded by the proposed Lyapunov function. 
Also, Theorem \ref{Theorem:5.1} guarantees the condition \eqref{Abar'PbarAbar-Pbar < -cPbar}, which is assumed in Proposition \ref{proposition:a<abar => boundedness}. 

Note that we specify the modal matrix $W_m$ in \eqref{eqn:switched} as the most delayed case -- all PEs use the oldest value in the buffer. Therefore, it can be inferred that $\lambda_{max}(P_m) \geq \lambda_{max}(P_j)$, $\forall j$, which results in 
\begin{align}
|| \bar{e}(k)||^2 < K\Big(1-\dfrac{1}{\lambda_{max}(P_m)}\Big)^k ||e(0)||^2,\label{eqn:22-1}
\end{align}
where $K$ is a positive constant. 
Therefore, the \textit{only information} required to compute the convergence rate of $||\bar{e}(k)||$, is the matrix $W_m$ with the corresponding positive definite matrix $P_m$. As a result, the rate of convergence can be calculated by the proposed methods \textit{without any scalability problems}.
\end{remark}

\section{Error Analysis}
In this section, we investigate the error probability, which quantifies the deviation of the random vector $X(k)$ from its steady-state value $X_{ss}$ in probability. To measure this error probability, the Markov inequality given by $\mathbf{Pr}\big(X\geq \epsilon\big) \leq \dfrac{\mathbb{E}[X]}{\epsilon}$, where $X$ is a nonnegative random variable and $\epsilon$ is a positive constant, is used.
First of all, we investigate the term $\vect\left(e(k)e(k)^{\top}\right)$ as follows:
\begin{align}
&\vect\left(e(k)e(k)^{\top}\right) =  \vect\left(\tilde{W}_{\sigma_{k-1}}e(k-1)e(k-1)^{\top}\tilde{W}_{\sigma_{k-1}}^{\top}\right)\nonumber\\
&\qquad = \big(\tilde{W}_{\sigma_{k-1}}\otimes \tilde{W}_{\sigma_{k-1}}\big)\vect\big(e(k-1)e(k-1)^{\top}\big).\label{eqn:21}
\end{align}
In the second equality of above equation, we used the property that $\vect(ABC) = (C^{\top}\otimes A)\vect(B)$.

By taking the expectation with new definitions $y(k) \triangleq \vect\left( e(k)e(k)^{\top}\right)$, $\bar{y}(k) \triangleq \mathbb{E}[y(k)]$, and $\Gamma_{\sigma_{k}}\triangleq \tilde{W}_{\sigma_k}\otimes \tilde{W}_{\sigma_k}$, \eqref{eqn:21} becomes
\begin{align*}
\bar{y}(k) &\triangleq \mathbb{E}[y(k)] = \mathbb{E}\left[\Gamma_{\sigma_{k-1}}y(k-1)\right]\\
&= \sum_{r=1}^{m}\mathbb{E}\left[ \Gamma_{\sigma_{k-1}} y(k-1) \:\Big|\: \sigma_{k-1}=r \right]\mathbf{Pr}(\sigma_{k-1}=r)\\
&= \sum_{r=1}^{m} \pi_r\Gamma_r \mathbb{E}\left[ y(k-1)\right],
\end{align*}
resulting in $\bar{y}(k) = \left(\sum_{r=1}^{m} \pi_r\Gamma_{r}\right) \bar{y}(k-1)$, where in the second line we applied the law of total probability and the last equality holds by $\mathbf{Pr}(\sigma_{k-1}=r) = \pi_r$ for i.i.d. switching.

By the exactly same argument given in Lemma \ref{lemma:5.1} and Proposition \ref{proposition:a<abar => boundedness}, the upper bound for $\bar{y}(k)$ is obtained as follows:
{\small
\begin{align}
||\bar{y}(k)|| < K\bigg(1-\dfrac{1}{\lambda_{max}(\tilde{P}_m)}\bigg)^{k/2} ||y(0)||, \quad \forall k\in\mathbb{N},\label{eqn:18-1}
\end{align}}
where $K$ is some positive constant and $\tilde{P}_m$ is a positive definite matrix, satisfying the condition $\Gamma_m^{\top}\tilde{P}_m\Gamma_m - \tilde{P}_m = -I$. 
However, unlike the positive definite matrix $P_m\in\mathbb{R}^{Nnq\times Nnq}$ in \eqref{eqn:14}, the dimension of the matrix $\tilde{P}_m$ is given by $\tilde{P}_m\in\mathbb{R}^{(Nnq)^2\times (Nnq)^2}$, which may be large in size, and hence incurs computational intractabilities to obtain such a $\tilde{P}_m$. 
Therefore, we introduce the following proposition and theorem in order to further facilitate the computation of $\lambda_{max}(\tilde{P}_m)$ as follows.

\begin{proposition}\label{proposition:4.2}
Consider a positive definite matrix $\tilde{P}_m$, satisfying the condition $\Gamma_m^{\top}\tilde{P}_m\Gamma_m - \tilde{P}_m = -I$, where $\Gamma_m \triangleq \tilde{W}_m\otimes \tilde{W}_m$, and $\tilde{W}_m$ is any real square matrix.
If we assume that there exist finite, positive constants $k_0$, $c_0$, and $c_1$ such that
\begin{eqnarray}
&1 \,\leq\, || \tilde{W}_m^{k} ||^4 \,\leq\, c_0,&\quad \text{for } k\in[0,k_0), \label{eqn:15-1}\\
&|| \tilde{W}_m^{k} ||^4 \:\:\leq c_1\: < 1,&\quad \text{for } k\in[k_0,\infty),\label{eqn:15-2}
\end{eqnarray}
then, the largest eigenvalue of $\tilde{P}_m$ is bounded by the following function:
\begin{align}
\lambda_{max}(\tilde{P}_m) < \sum_{k=0}^{\infty}|| \tilde{W}_m^k ||^4 \leq k_0c_0\left(\dfrac{1}{1-c_1}\right),\label{eqn:15}
\end{align}
\end{proposition}

\begin{proof}
The leftmost inequality in \eqref{eqn:15} can be proved as follows. The positive definite matrix $\tilde{P}_m$ satisfying the condition $\Gamma_m^{\top}\tilde{P}_m\Gamma_m - \tilde{P}_m = -I$, is analytically computed by $\tilde{P}_m = \sum_{k=0}^{\infty}\left({\Gamma_m^{\top}}^k\right)I\left(\Gamma_m^k\right) = \sum_{k=0}^{\infty}{\Gamma_m^{\top}}^k\Gamma_m^k$. 
Then, for a given matrix $\Gamma_m\triangleq \tilde{W}_m\otimes \tilde{W}_m$, we have 
\begin{align}
&{\Gamma_m^{\top}}^k\Gamma_m^k 
< \rho({\Gamma_m^{\top}}^{k}\Gamma_m^k)I = \rho({\Gamma_m^k}^{\top}\Gamma_m^k)I 
= \sigma_{max}^2(\Gamma_m^k)I\nonumber\\
&\qquad\:= || \Gamma_m^k||^2I
= || (\tilde{W}_m\otimes \tilde{W}_m)^k ||^2I
= || \tilde{W}_m^{k} ||^4I,\label{eqn:22}
\end{align}
where $\rho(\cdot)$ and $\sigma_{max}(\cdot)$ denote the spectral radius and the spectral norm, respectively.
For equality conditions in \eqref{eqn:22}, we used the known property that $\sqrt{\rho({\Gamma_m^k}^{\top}\Gamma_m^k)} = \sigma_{max}(\Gamma_m^k) = ||\Gamma_m^k||$ and $||(\tilde{W}_m\otimes \tilde{W}_m)^k|| = ||\tilde{W}_m^k\otimes \tilde{W}_m^k|| = ||\tilde{W}_m^k||^2$, $\forall k\in\mathbb{N}_0$.
By summing up from $k=0$ to $\infty$, and then taking the largest eigenvalue in \eqref{eqn:22}, we have $\lambda_{max}(\tilde{P}_m) = \lambda_{max}\left(\sum_{k=0}^{\infty}{\Gamma^{\top}}^k\Gamma^k\right) < \sum_{k=0}^{\infty}||\tilde{W}_m^k||^4$.

For the rightmost inequality in \eqref{eqn:15}, the assumptions in \eqref{eqn:15-1}-\eqref{eqn:15-2} result in
\begin{align*}
&\sum_{k=0}^{\infty}||\tilde{W}_m^k||^4 = \underbrace{\sum_{k=0}^{k_0-1}||\tilde{W}_m^k||^4}_{\leq k_0c_0} + \sum_{k=k_0}^{\infty}||\tilde{W}_m^k||^4\\
&\leq k_0c_0 + \sum_{k=k_0}^{2k_0-1}||\tilde{W}_m^k||^4 + \sum_{\substack{k=2k_0}}^{3k_0-1}||\tilde{W}_m^k||^4 + \cdots\\
&= k_0c_0 + \sum_{k=0}^{k_0-1}||\tilde{W}_m^{(k_0+k)}||^4 + \sum_{\substack{k=2k_0}}^{3k_0-1}||\tilde{W}_m^k||^4 +\cdots\\
&\leq k_0c_0 + \underbrace{||\tilde{W}_m^{k_0}||^4}_{\leq c_1}\underbrace{\sum_{k=0}^{k_0-1}||\tilde{W}_m^{k}||^4}_{\leq k_0c_0} + \sum_{k=0}^{k_0-1}||\tilde{W}_m^{(2k_0+k)}||^4 +\cdots\\
&\leq k_0c_0 + k_0c_0c_1 + \underbrace{||\tilde{W}_m^{2k_0}||^4}_{\leq c_1^2}\underbrace{\sum_{k=0}^{k_0-1}||\tilde{W}_m^{k}||^4}_{\leq k_0c_0} +\cdots\\
&\leq k_0c_0 + k_0c_0c_1 + k_0c_0c_1^2 + \cdots\\
&= k_0c_0\left(\sum_{n=0}^{\infty}c_1^n\right) = k_0c_0\left(\dfrac{1}{1-c_1}\right).
\end{align*}
Hence, we have $\displaystyle\sum_{k=0}^{\infty}|| \tilde{W}_m^k ||^4 \leq k_0c_0\left(\dfrac{1}{1-c_1}\right)$.
\end{proof}

\begin{theorem}\label{theorem:4.1}
Consider a \textit{stable}, i.i.d. jump linear system with subsystem dynamics $\tilde{W}_j$ given in \eqref{eqn:8}. 
Then, the probability of $||e(k)||^2 > \epsilon$, where $\epsilon$ is some positive constant, is given by
\begin{align}
\textbf{Pr}\bigg(||e(k)||^2 > \epsilon \bigg) \leq \min(1,\beta), \quad k\in\mathbb{N}_0,\label{eqn:19}
\end{align}
where $\beta \triangleq \dfrac{\sqrt{n}K}{\epsilon}\Bigg(1 - \dfrac{1-c_1}{k_0c_0}\Bigg)^{k/2}||y(0)||$, $K>0$ is a constant, $c_0,c_1,k_0$ are positive constants such that the conditions \eqref{eqn:15-1}-\eqref{eqn:15-2} are satisfied.
\end{theorem}

\begin{proof} At first, we consider the following equality condition given by 
\begin{align}
&||e(k)||^2 = e(k)^{\top}e(k) = \tr(e(k)^{\top}e(k)) = \tr\big(I\left(e(k)e(k)^{\top}\right)\big)\nonumber\\
&= \vect(I)^{\top}\vect(e(k)e(k)^{\top}) = \vect(I)^{\top}y(k) ,\label{eqn:29}
\end{align}
where we used the cyclic permutation property for the trace operator in the first line and  the equality in the second line holds by the property $\tr(X^{\top}Y) = \vect(X)^{\top}\vect(Y)$ for any square matrix $X,Y\in\mathbb{R}^{n\times n}$.

We take the expectation in both sides of \eqref{eqn:29}, which leads to
\begin{align}
\mathbb{E}\big[||e(k)||^2\big] &=  \vect(I)^{\top} \mathbb{E}\big[y(k)\big] = \vect(I)^{\top}\bar{y}(k).\label{eqn:30}
\end{align}

Since the term $\mathbb{E}\big[||e(k)||^2\big]$ is a scalar value, taking the Euclidean norm returns the same value. Hence, applying the Euclidean norm in \eqref{eqn:30} results in
\begin{align}
\mathbb{E}\big[||e(k)||^2\big] &= || \vect(I)^{\top}\bar{y}(k) ||\nonumber\\
&\leq ||\vect(I)^{\top}||\cdot ||\bar{y}(k)|| = \sqrt{n}\cdot ||\bar{y}(k)||\label{eqn:31}.
\end{align}

Now, plugging \eqref{eqn:18-1} and \eqref{eqn:15} into \eqref{eqn:31} leads to
\begin{align*}
\mathbb{E}\big[||e(k)||^2\big] < \sqrt{n}K\Bigg(1 - \dfrac{1-c_1}{k_0c_0}\Bigg)^{k/2}||y(0)||.
\end{align*}
Finally, by applying the Markov inequality the above equation ends up with
\begin{align*}
\textbf{Pr}\bigg(||e(k)||^2 > \epsilon \bigg) &\leq \dfrac{\mathbb{E}\big[||e(k)||^2\big]}{\epsilon} < \beta,
\end{align*}
where $\beta \triangleq \dfrac{\sqrt{n}K}{\epsilon}\Bigg(1 - \dfrac{1-c_1}{k_0c_0}\Bigg)^{k/2}||y(0)||$.

Since the probability cannot exceed one, we have $\textbf{Pr}\bigg(||e(k)||^2 > \epsilon \bigg) \leq \min(1,\beta)$
\end{proof}

Theorem \ref{theorem:4.1} represents the error probability for a given bound $\epsilon$. Since $e(k)$ is a time-varying variable, the probability $\textbf{Pr}\left(||e(k)||^2>\epsilon\right)$ also changes with respect to time. Starting from a given initial condition $y(0)$, this probability will converge to zero if $\Bigg(1 - \dfrac{1-c_1}{k_0c_0}\Bigg)<1$.

\section{Simulations}
In order to test the proposed methods, simulation was carried out for the one-dimensional heat equation. We implemented the asynchronous parallel algorithm with $\mathtt{CUDA\:\:C\text{++}\:\:programming}$ on $\mathtt{nVIDIA\: Tesla^{{\scriptsize{TM}}}\: C2050}$ GPU, which has $448$ $\mathtt{CUDA\: cores}$. 
The simulations were performed with the following parameters:
\begin{itemize}
\item{Simulation Parameters:}
\begin{eqnarray*}
&\Delta x =& 0.1, \Delta t = 0.01, \alpha = 0.5, r = \alpha\dfrac{\Delta t}{\Delta x^2} = 0.5\\
&I.C.:& u_i = \text{cos}^2\bigg( \dfrac{3\pi i}{2(N-1)} \bigg), \:i=1,2,\hdots,N\\
&B.C.:& u_1(k) = 1,\: u_N(k) = 0, \: \forall k
\end{eqnarray*}
\item{Buffer length:} $q=3$
\item{Number of PEs:} $N=100$.
\item{Number of grid points in PE:} $n=1$\\
\end{itemize}

\begin{figure}
\centering
\includegraphics[scale=0.2]{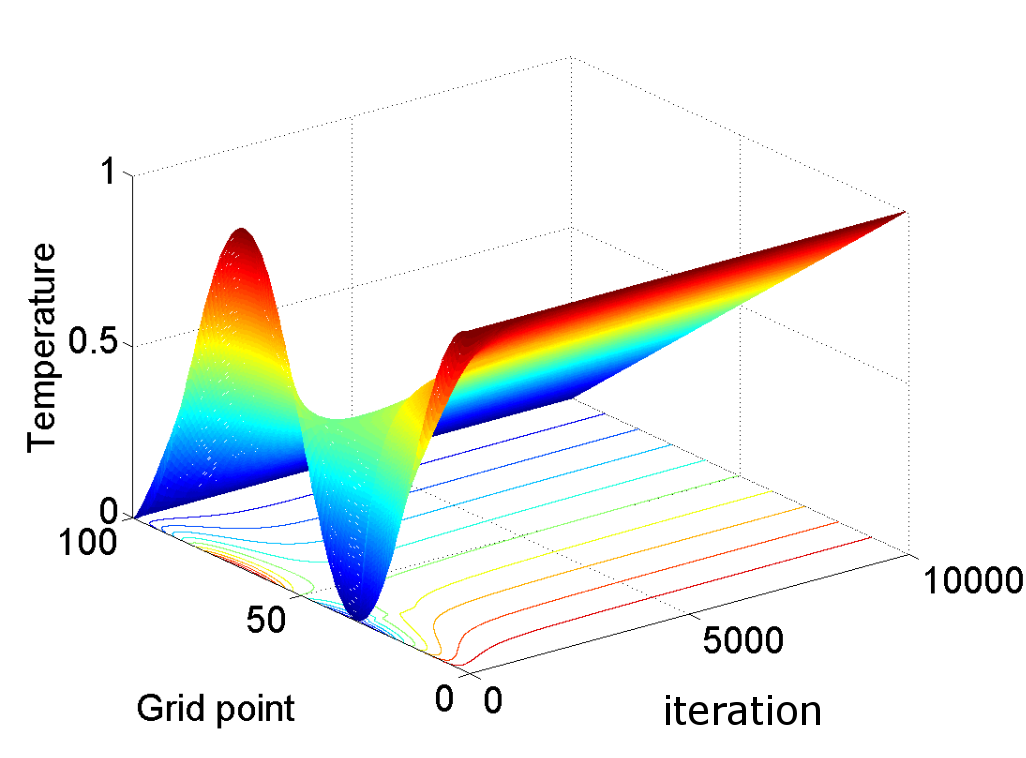}
\caption{The spatio-temporal change of the temperature. Initially, the temperature was given by the cosine square function. The total grid points are $100$, and the simulation was terminated when $k=10000$.}\label{fig.2}\vspace{-0.1in}
\end{figure}

For a given initial temperature, the spatio-temporal evolution of the state is presented in Fig. \ref{fig.2}. As time $k$ increases, the curved shape of the temperature, given as a cosine square function initially, flattens out. This simulation represents the synchronous case.

\begin{figure}
\centering
\subfigure[Stability]{\includegraphics[scale=0.29]{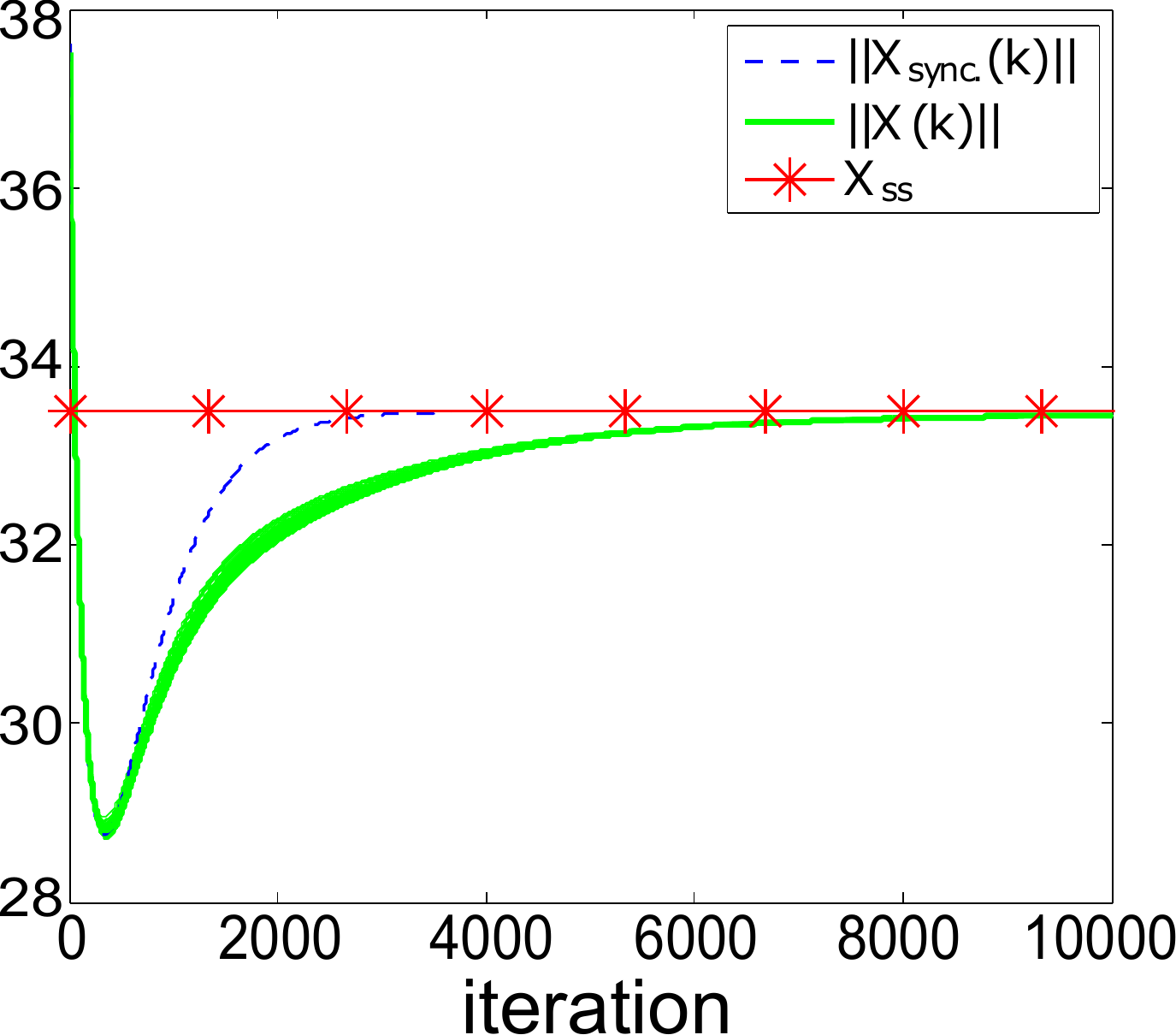}}
\subfigure[Convergence rate]{\includegraphics[scale=0.34]{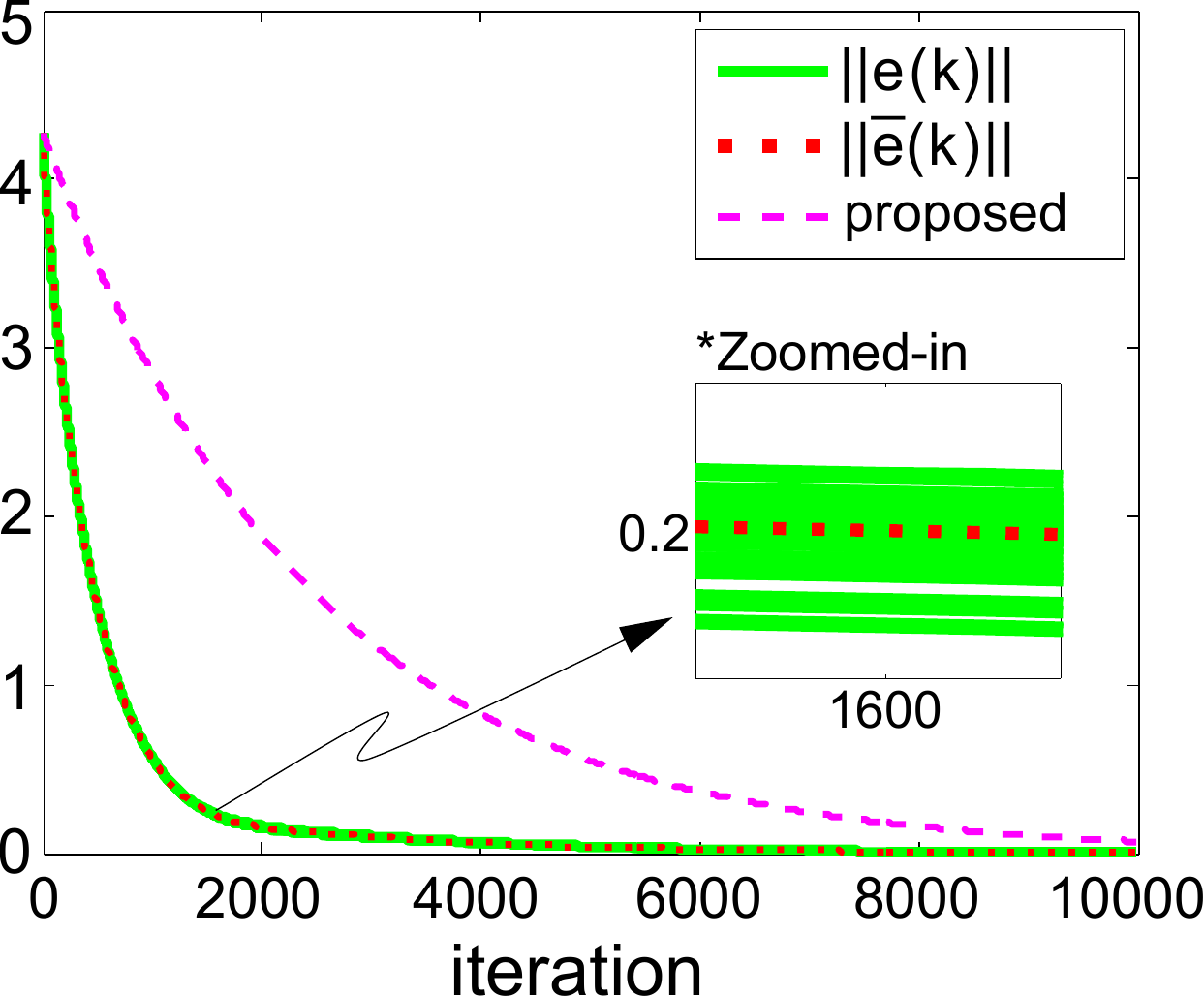}}
\caption{
The results for the stability and convergence rate. (a) The solid lines represent the ensembles of total $300$ simulations. The synchronous case is given by dashed line. The steady-state is depicted by starred line. (b) The solid and dotted lines represent 300 ensembles for $||e(k)||$ and the normed empirical mean $||\bar{e}(k)||$, respectively. The dashed line shows the upper bound of $||\bar{e}(k)||$ from the proposed Lyapunov function, respectively.}\label{fig.3}\vspace{-0.15in}
\end{figure}

In Fig. \ref{fig.3} (a), the ensemble of the trajectories is shown for the asynchronous algorithm. The solid lines show the trajectories of total $300$ simulations. Due to the randomness in the asynchronous algorithm, the trajectories differ from each other. For a reference, the synchronous scheme is also shown by a dashed line. Although it seems that the synchronous scheme converges faster with respect to the given iteration step, the physical simulation time may take more because the idle time is necessary at each iteration in the synchronous case.
As the proposed method guarantees the stability through the common eigenvectors, both synchronous and asynchronous trajectories converged to the same steady-state value $X_{ss}$, depicted by starred line. 

Next, we present the result for the convergence rate of the asynchronous algorithm. 
We assume that the switching probability $\Pi$ has the form of an i.i.d. jump process. Fig. \ref{fig.3} (b) shows the convergence rate of $||\bar{e}(k)||$, which describes how fast the expected value of the state converges to $X_{ss}$.
The solid lines are 300 sample trajectories of $||e(k)||$, starting from the given initial condition: $e(0) = X(0) - X_{ss}$.
The dotted line depicts the time history of the normed \textit{empirical} mean $||\bar{e}(k)||$, whereas the dashed line shows an upper bound by the proposed Lyapunov method \eqref{eqn:22-1}. Note that $||e(k)||$ is a random variable, and hence the normed empirical mean $||\bar{e}(k)||$ was obtained by averaging the data over $300$ simulations. In the proposed method, however, it is not necessary to execute the simulation multiple times. 

Fig. \ref{fig.4} represents the result for the error probability with respect to time and $\epsilon$. 
For different values of $\epsilon$, Fig. \ref{fig.4} (a) and (b) describe the time history of the error probabilities.
The solid line denotes the empirical probability obtained from data -- i.e., the number of samples satisfying $||e(k)||^2>\epsilon$ divided by the total number of samples. The dashed line depicts the Markov inequality, computed from $\frac{\mathbb{E}\left[||e(k)||^2\right]}{\epsilon}$, where $\mathbb{E}\left[||e(k)||^2\right]$ is obtained by the statistics. Finally, the cross symbols mean the upper bound by the proposed method.
As shown in Fig. \ref{fig.4} (a) and (b), the probabilities for all cases converge to zero since the error is asymptotically convergent.

\begin{figure}
\begin{center}
\subfigure[$\epsilon = 0.01$]{\includegraphics[scale=0.3]{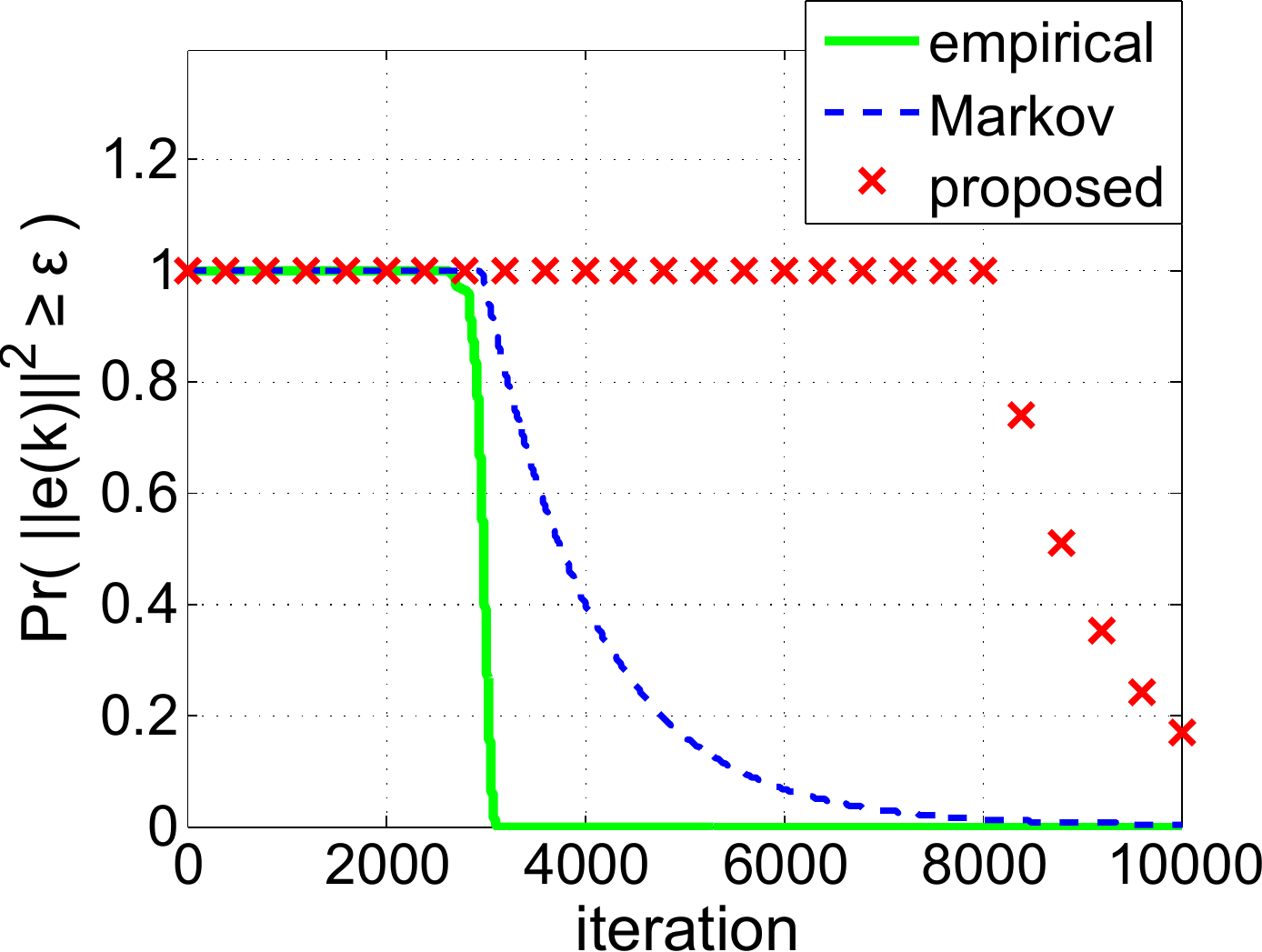}}
\subfigure[$\epsilon = 1$]{\includegraphics[scale=0.3]{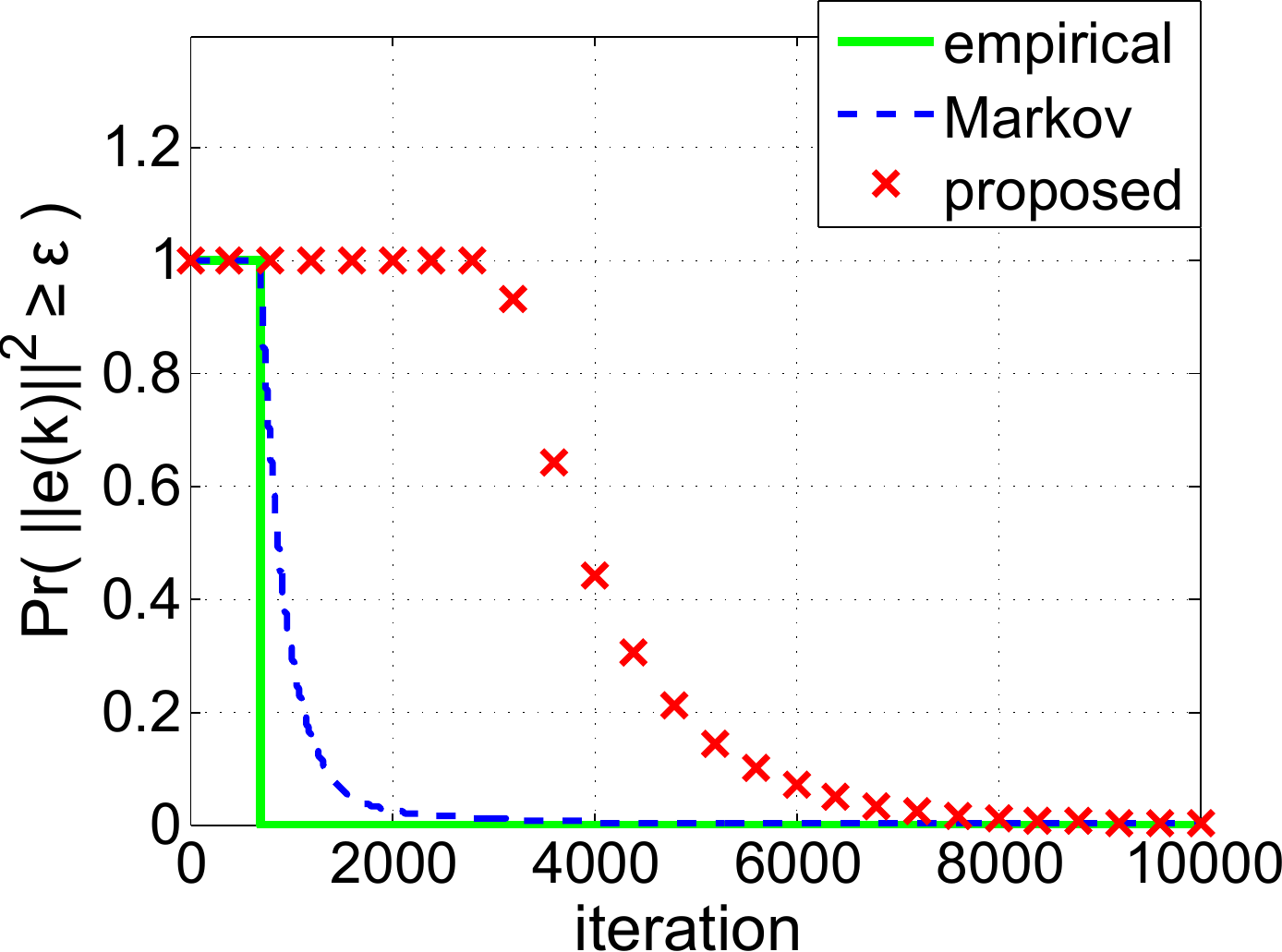}}
\subfigure[$T = 2\times 10^{-6}$]{\includegraphics[scale=0.3]{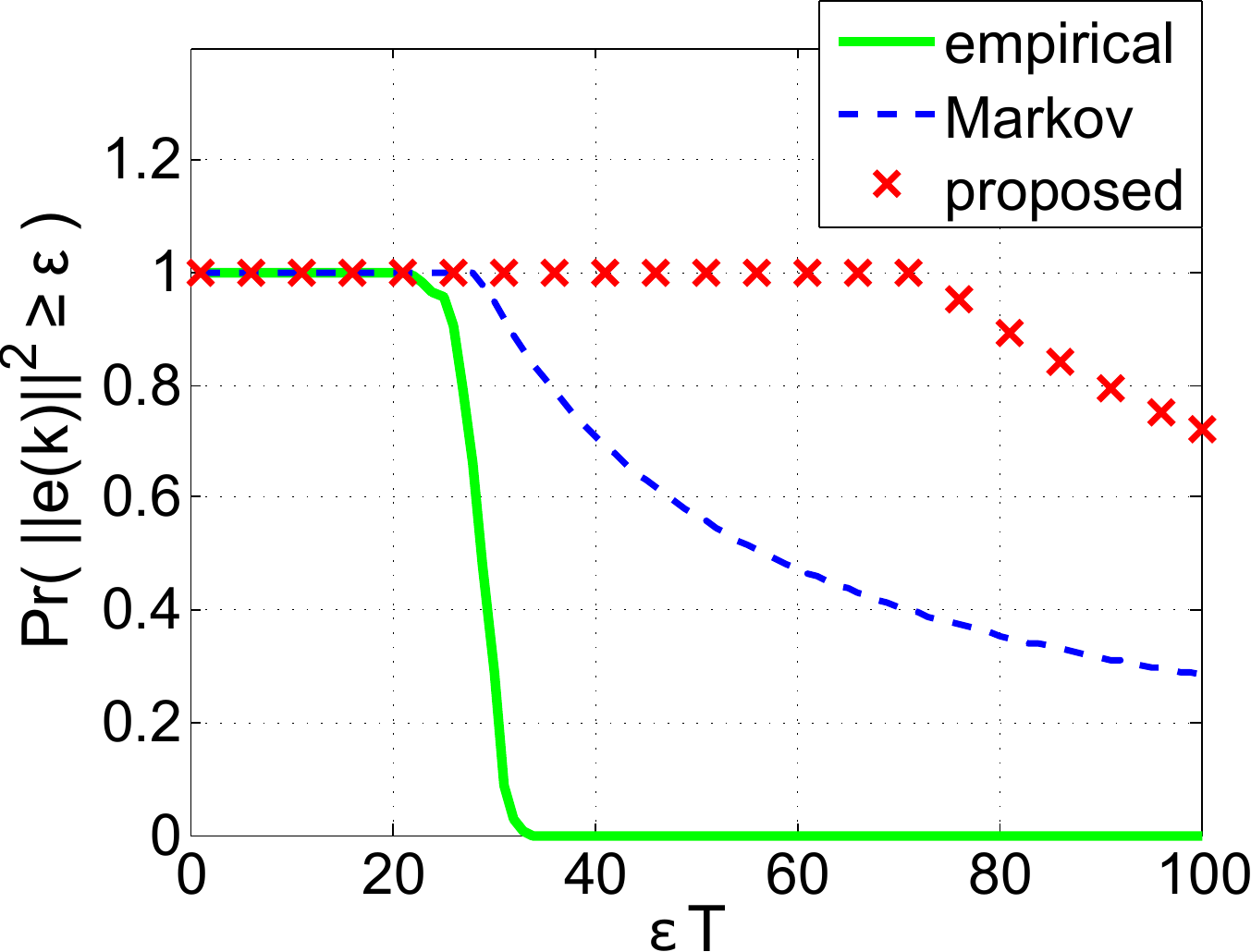}}
\subfigure[$T = 12\times 10^{-6}$]{\includegraphics[scale=0.3]{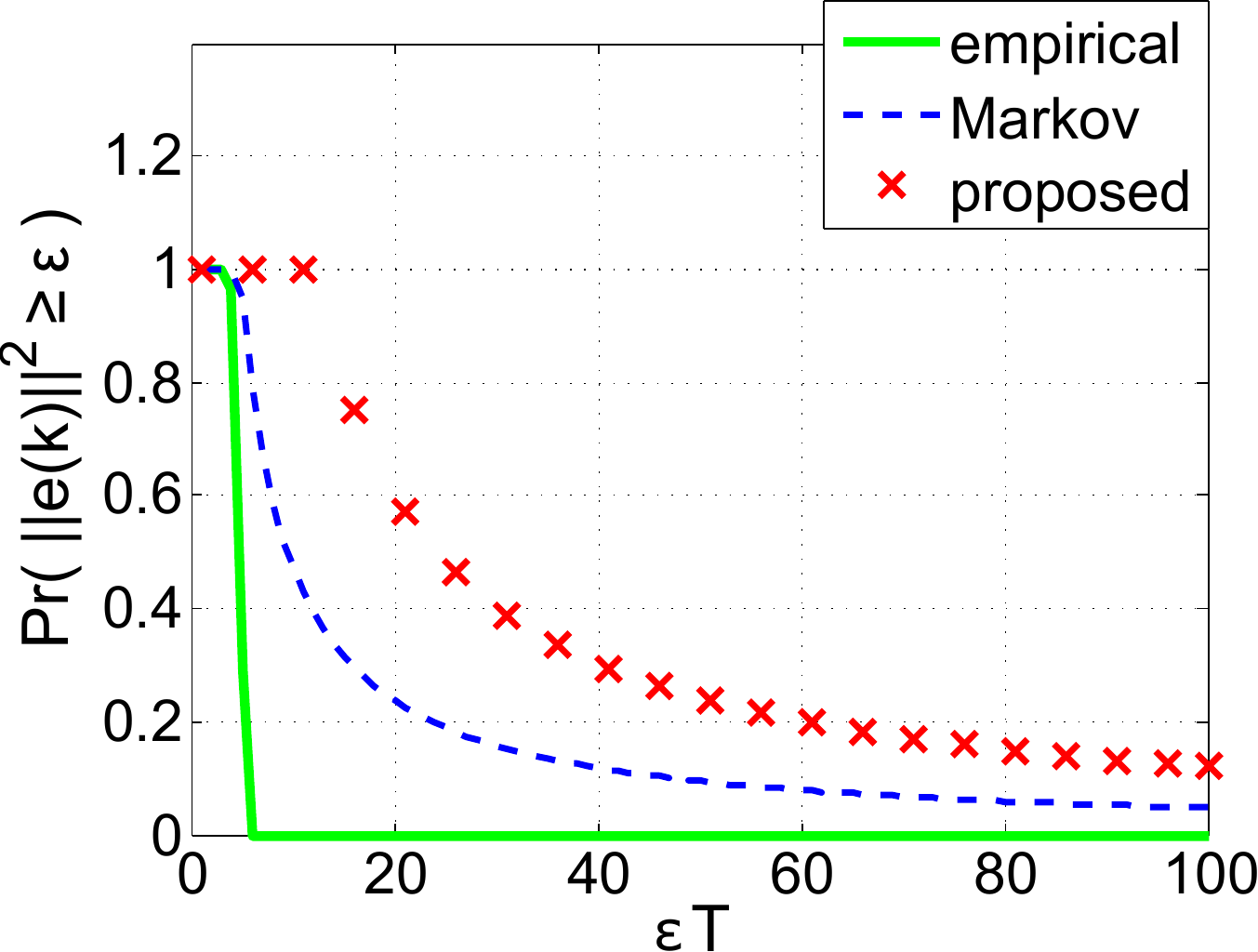}}
\caption{Error probability with respect to time (a), (b) and with respect to $\epsilon$ (c), (d). The solid line and dashed line represent empirical error probability and empirical Markov inequality, respectively. The cross symbol denotes the upper bound for the error probability by the proposed method.} \label{fig.4}\vspace{-0.15in}
\end{center}
\end{figure}

On the other hands, Fig. \ref{fig.4} (c), (d) show the error probability with respect to $\epsilon$ at fixed time instance. In this result, the time is fixed at $k=9000$ out of total $10000$ iteration times, and the probability is computed while increasing $\epsilon$ values. In Fig. \ref{fig.4} (c) and (d), $\epsilon T$ is given by the index along $x$-axis, where the value of $T$ is given in Fig. \ref{fig.4} (c) and (d), respectively.
In both cases, the error probabilities decrease as $\epsilon$ increases.

Although the proposed methods provide a conservative bound, it does not require executing the code multiple times to predict the convergence rate or the error probability. In addition to that, the proposed methods are carried out in a computationally efficient manner without storing all subsystem matrices. In this example, we have  $m=3^{2(100-2)}\approx 3^{200}$, and keeping $3^{200}$ numbers of matrices is intractable in the real implementation. The proposed method, however, guarantees the convergence rate and the error probability, without any scalability issues.  Therefore, the presented methods provide a computationally efficient tool to analyze the asynchronous numerical schemes.

\section{Conclusions}
This paper studied the stability, convergence rate, and error probability of the asynchronous parallel numerical algorithm. The asynchronous algorithm achieves better performance in terms of the total simulation time, particularly when massively parallel computing is required because it doesn't wait for synchronization across PEs. 
In order to analyze the asynchronous numerical algorithm, we adopted the switched linear system framework. Although modeling of massively parallel numerical algorithms as switched dynamical systems results in a very large number of modes, we developed new methods that circumvent this scalability issue.
While the results presented here are based on 1D heat equation, the analysis approach is generic and be applicable to other PDEs as well.

\bibliographystyle{ieeetr}
\bibliography{ACC2015_Stability_bib}

\begin{thebibliography}{10}

\bibitem{fan2004gpu}
Z.~Fan, F.~Qiu, A.~Kaufman, and S.~Yoakum-Stover, ``Gpu cluster for high
  performance computing,'' in {\em Proceedings of the 2004 ACM/IEEE conference
  on Supercomputing}, p.~47, IEEE Computer Society, 2004.

\bibitem{owens2008gpu}
J.~D. Owens, M.~Houston, D.~Luebke, S.~Green, J.~E. Stone, and J.~C. Phillips,
  ``Gpu computing,'' {\em Proceedings of the IEEE}, vol.~96, no.~5,
  pp.~879--899, 2008.

\bibitem{nickolls2008scalable}
J.~Nickolls, I.~Buck, M.~Garland, and K.~Skadron, ``Scalable parallel
  programming with cuda,'' {\em Queue}, vol.~6, no.~2, pp.~40--53, 2008.

\bibitem{donzis2014asynchronous}
D.~A. Donzis and K.~Aditya, ``Asynchronous finite-difference schemes for
  partial differential equations,'' {\em Journal of Computational Physics},
  vol.~274, pp.~370--392, 2014.

\bibitem{daafouz2002stability}
J.~Daafouz, P.~Riedinger, and C.~Iung, ``Stability analysis and control
  synthesis for switched systems: a switched lyapunov function approach,'' {\em
  Automatic Control, IEEE Transactions on}, vol.~47, no.~11, pp.~1883--1887,
  2002.

\bibitem{lin2005stability}
H.~Lin, G.~Zhai, L.~Fang, and P.~J. Antsaklis, ``Stability and h∞ performance
  preserving scheduling policy for networked control systems,'' in {\em Proc.
  16th IFAC World Congress on Automatic Control}, 2005.

\bibitem{lee2014optimal}
K.~Lee and R.~Bhattacharya, ``Optimal switching synthesis for jump linear
  systems with gaussian initial state uncertainty,'' in {\em ASME 2014 Dynamic
  Systems and Control Conference}, pp.~V002T24A003--V002T24A003, American
  Society of Mechanical Engineers, 2014.

\bibitem{hassibi1999control}
A.~Hassibi, S.~P. Boyd, and J.~P. How, ``Control of asynchronous dynamical
  systems with rate constraints on events,'' in {\em Decision and Control,
  1999. Proceedings of the 38th IEEE Conference on}, vol.~2, pp.~1345--1351,
  IEEE, 1999.

\bibitem{xiao2000control}
L.~Xiao, A.~Hassibi, and J.~P. How, ``Control with random communication delays
  via a discrete-time jump system approach,'' in {\em American Control
  Conference (ACC), 2000. Proceedings of the 2000}, vol.~3, pp.~2199--2204,
  IEEE, 2000.

\bibitem{zhang2005new}
L.~Zhang, Y.~Shi, T.~Chen, and B.~Huang, ``A new method for stabilization of
  networked control systems with random delays,'' {\em Automatic Control, IEEE
  Transactions on}, vol.~50, no.~8, pp.~1177--1181, 2005.

\bibitem{liu2009stabilization}
M.~Liu, D.~W. Ho, and Y.~Niu, ``Stabilization of markovian jump linear system
  over networks with random communication delay,'' {\em Automatica}, vol.~45,
  no.~2, pp.~416--421, 2009.

\bibitem{lee2014acc}
K.~Lee, A.~Halder, and R.~Bhattacharya, ``Probabilistic robustness analysis for
  stochastic jump linear systems,'' in {\em American Control Conference (ACC),
  2014. Proceedings of the 2014}, pp.~2638--2643, IEEE, 2014.

\bibitem{lee2015performance}
K.~Lee, A.~Halder, and R.~Bhattacharya, ``Performance and robustness analysis
  of stochastic jump linear systems using wasserstein metric,'' {\em
  Automatica}, vol.~51, pp.~341--347, 2015.

\bibitem{lee2015DNCS}
K.~Lee and R.~Bhattacharya, ``Stability analysis of large-scale distributed
  networked control systems with random communication delays: A switched system
  approach,'' {\em arXiv preprint arXiv:1503.03047}, 2015.

\bibitem{pletcher2012computational}
R.~H. Pletcher, J.~C. Tannehill, and D.~Anderson, {\em Computational fluid
  mechanics and heat transfer}.
\newblock CRC Press, 2012.

\bibitem{laffey2007tensor}
T.~J. Laffey and H.~{\v{S}}migoc, ``Tensor conditions for the existence of a
  common solution to the lyapunov equation,'' {\em Linear algebra and its
  applications}, vol.~420, no.~2, pp.~672--685, 2007.

\bibitem{shorten2003result}
R.~Shorten, K.~S. Narendra, and O.~Mason, ``A result on common quadratic
  lyapunov functions,'' {\em Automatic Control, IEEE Transactions on}, vol.~48,
  no.~1, pp.~110--113, 2003.

\bibitem{liberzon2003switching}
D.~Liberzon, {\em Switching in systems and control}.
\newblock Springer, 2003.

\bibitem{gurvits2007stability}
L.~Gurvits, R.~Shorten, and O.~Mason, ``On the stability of switched positive
  linear systems,'' {\em Automatic Control, IEEE Transactions on}, vol.~52,
  no.~6, pp.~1099--1103, 2007.

\bibitem{lin2009stability}
H.~Lin and P.~J. Antsaklis, ``Stability and stabilizability of switched linear
  systems: a survey of recent results,'' {\em Automatic control, IEEE
  Transactions on}, vol.~54, no.~2, pp.~308--322, 2009.

\end{thebibliography}
\end{document}